\documentclass[12pt]{article}

\usepackage[T1]{fontenc}
\usepackage{graphicx}
\usepackage[latin9]{inputenc}
\usepackage{amsmath,mathrsfs,mathtools,amsthm}
\usepackage{amsmath,amsfonts,amssymb,epsfig}
\usepackage[usenames,dvipsnames]{color}

\usepackage{float}
\usepackage{array}

\newtheorem{definition}{Definition}
\newtheorem{theorem}{Theorem}
\newtheorem{proposition}{Proposition}
\newtheorem{lem}{Lemma}
\newtheorem{cor}{Corollary}
\newtheorem{remark}{Remark}
\newtheorem{example}{Example}

\newcommand{\C}{\mbox{$\cal C$}}

\newcommand{\ben}{\begin{equation*}}
\newcommand{\een}{\end{equation*}}

\newcommand{\comment}[1]{}

\newcommand{\F}{\mathbb{F}}

\newcommand{\A}{\mbox{$\mathcal A$}}

\newcommand{\x}{\mbox{\bf x}}
\newcommand{\y}{\mbox{\bf y}}

\newcommand{\z}{\mbox{\bf z}} 

\begin{document}

 \author{ Jon-Lark Kim\thanks{jlkim@sogang.ac.kr,
Jon-Lark Kim is with Department of Mathematics, Sogang University, Seoul, South Korea.}}

\title{Fuzzy linear codes based on nested linear codes (Revised on Sep. 8, 2024)
}


\date{}

\maketitle

\begin{abstract}
In this paper, we describe a correspondence between a fuzzy linear code and a family of nested linear codes. We also describe the arithmetic of fuzzy linear codes.
As a special class of nested linear codes, we consider a family of nested self-orthogonal codes.
  A linear code is self-orthogonal if it is contained in its dual and self-dual if it is equal to its dual. We introduce a definition of fuzzy self-dual or self-orthogonal codes which include classical self-dual or self-orthogonal codes.
   As examples, we construct several interesting classes of fuzzy linear codes including fuzzy Hamming codes, fuzzy Golay codes, and fuzzy Reed-Muller codes. We also give a general decoding algorithm for fuzzy linear codes.

\end{abstract}

\noindent
{\bf keywords} error-correcting code, fuzzy code, fuzzy set\\


\section{Introduction}

{\color{blue}{(I have received some comments on the journal version of the paper. So in this revision, I corrected some errors in blue color. I appreciate all for those comments.)}}

For the past half century, coding theory which originated from Shannon's seminal paper~\cite{Sha} has grown
into a discipline intersecting mathematics and engineering with applications to almost every
area of communication such as satellite and cellular telephone transmission, compact disc
recording, and data storage~\cite{HuffmanPless}.
Shannon's results guarantee that the data can be encoded before transmission so that the altered data can be decoded to the specified degree of accuracy~\cite{HuffmanPless}.

In coding theory, the encoding process assigns each vector of length $k$, called information, to a unique vector of length $n$ ($k \le n$), called a codeword of length $n$ by adding the redundancy. In terms of fuzzy set theory, each codeword has membership value 1 and non-codewords have membership value 0. Therefore, any $k$-dimensional subspace $\mathcal C$, called a linear code, of a given $n$-dimensional space $\mathbb F_q^n$ over a finite field $\mathbb F_q$ can be regarded as a crisp set. The main assumption in coding theory is that the vectors of length $k$ as well as their corresponding codewords appear equally likely in the communication channel. However, in a more general situation, it is natural to think that vectors of a certain length are encoded into fuzzy codewords with different membership values. For example, suppose that we have four information $0, 1,2,3$ in $\mathbb Z_4$ and encode them into $(00), (01), (11), (10)$ in $\mathbb F_2^2$ respectively, called the Gray map, with the corresponding membership values $1, 1/2, 1/4, 1/4$ in this order. This defines a fuzzy set $A$ in $\mathbb F_2^2$. It is interesting to see that the three upper $\alpha$-level cuts $A_1=\{ (00) \}, A_{1/2}=\{ (00), (01) \}, A_{1/4}=\mathbb F_2^2$ are all linear codes. Therefore, fuzzy set theory provides more general concepts and tools in the study of classical coding theory. This is one of the roles of fuzzy set theory in this work.

Hence, an interesting and challenging problem is to send fuzzy codewords over a noisy channel and decode/correct them if there are errors in the received vectors. In order to do this, it is natural to apply classical coding theory to fuzzy codewords. First of all, we need to construct interesting fuzzy codewords (construction problem), and then we need to decode altered fuzzy codewords efficiently (decoding problem).
Fuzzy linear codes, that is, fuzzy linear spaces are a reasonable choice since they can be simply described using the concept of generator matrices. This motivates the study of fuzzy linear codes.

More precisely, a fuzzy linear code naturally defines a family of nested linear codes using the concept of
the upper $\alpha$-level cuts. Conversely, any family of nested linear codes defines a fuzzy linear code whose upper $\alpha$-cuts correspond to nested linear codes. Moreover, a given linear code can produce a family of nested linear codes by considering a sequence of its nested subcodes, itself, and its nested supercodes. Therefore, a fuzzy linear code is a parameterized way to describe a linear code.
This is one of the advantages to study fuzzy linear codes because this approach gives a new merit to understand a single linear code in terms of its nested subcodes and nested supercodes. For example, the well known family of Reed-Muller codes in Section~\ref{sec-RM} is in fact a fuzzy linear code, which was not noticed in any literature to the best of our knowledge.

The first concept of a fuzzy linear space was introduced by Katsaras and Liu ~\cite{KatLiu} in 1977, followed by
 Nanda~\cite{Nan} and Biswas~\cite{Bis}. Shum and De Gang~\cite{ShuGan} defined a fuzzy linear code using the definition of a fuzzy linear space and studied the fuzzy cyclic codes over a finite field $\mathbb F_q$ using the fuzzy ideal of a certain group algebra. Their approach~\cite{ShuGan} mainly uses the concept of classical cyclic codes which are defined as an ideal of a certain quotient ring.

On the other hand, codes over finite rings including finite fields are related to module theory.
 The concept of fuzzy modules was introduced by Negoita and Ralescu~\cite{NegRal} while the notion of fuzzy submodule was introduced by
Maschinchi and Zahedi~\cite{MasZah}. Using this fuzzy module theory, Atamewoue et al.~\cite{Ata} studied fuzzy linear codes and fuzzy cyclic codes over $\mathbb Z_{p^k}$.

Nevertheless, more theories of fuzzy linear codes are still under development. In particular, we propose two main problems. The first problem is about the construction of a family of fuzzy linear codes other than the family of cyclic codes which were studied by~\cite{ShuGan} and~\cite{Ata}. The second problem is about an efficient decoding algorithm for fuzzy linear codes because the previously known papers~\cite{ShuGan},~\cite{Ata} describe only the basic ideas of decoding of fuzzy linear codes with no special advantage over a classical decoding.

To solve the first problem, we consider self-orthogonal/self-dual codes.
We note that
a given linear code $\C$ can define a fuzzy linear code which gives a family of nested linear codes $\C_0 \subset \C_1 \subset \cdots \subset \C_i =\C  \subset  \cdots \subset \C_{m-1}$ for some $i$ and $m$. One problem of this association is that there are many non-canonical choices of supercodes $\C_{j'}$ of $\C$ ($j' > i$) of $\C$. We can solve this problem if we start from a self-orthogonal code because $\C_j$ determines $\C_j'$ by $\C_j'=\C_j^{\perp}$. Hence, fuzzy self-orthogonal codes are simpler to describe than general fuzzy linear codes. This motives why fuzzy self-orthogonal or self-dual fuzzy codes are interesting.

We recall that a linear code is called self-orthogonal if it is contained in its dual under the Euclidean inner product and self-dual if it is equal to its dual.
It is well known that self-orthogonal/self-dual codes have many mathematical properties and have wide applications to lattices, designs, and quantum codes~\cite{RaiSlo}. This paper gives another application of self-orthogonal/self-dual codes to fuzzy set theory.

In this paper, we show that fuzzy self-orthogonal or  self-dual codes include several interesting classes of fuzzy linear codes, called fuzzy Hamming codes, fuzzy Golay codes, and fuzzy Reed-Muller codes.
We also give a general decoding algorithm for fuzzy linear codes which turns out to be efficient because it reduces the size of the syndrome table.

\medskip
This paper is organized as follows. Section 2 gives facts from coding theory and fuzzy linear codes. Section 3 describes the arithmetic of fuzzy linear codes. In Section 4, we give the definition of fuzzy self-dual codes and examples associated with the famous binary Hamming $[8,4,4]$ code and the binary Golay $[24,12,8]$ code. In Section 5, fuzzy self-orthogonal codes are introduced as a generalization of fuzzy self-dual codes. Section 6 shows that the binary Reed-Muller codes can be associated with fuzzy Reed-Muller codes, which are also self-orthogonal or self-dual. Section 7 describes the Syndrome Decoding Algorithm for fuzzy linear codes. Section 8 concludes the paper.

\section{Preliminaries}




In this section, we review basic facts from coding theory and fuzzy set theory. For coding theory, we refer to~\cite{CarMarManPic},~\cite{GurRudSud},~\cite{HuffmanPless},~\cite{MS}. We also refer to \cite{SyrGra} for basic fuzzy set theory and to ~\cite{Hoh} and~\cite{MaYan} for a recent theory and an application, respectively.

Let $\mathbb F_q$ be a finite field with $q$ elements, where $q=p^r$ for some prime $p$ and a positive integer $r \ge 1$. Let
$$\F_q^n = \{\x=(x_1, x_2, \dots, x_n) ~|~ x_i \in \F_q {\mbox{ for all }} i \}.$$
 Then $\F_q^n$ is an $n$-dimensional vector space over $\F_q$.
We recall that a subset $S$ of $\F_q^n$ is called a {\em linear space} or a {\em subspace} of $\F_q^n$ if $\x + \y  \in S$ for any $\x, \y \in S$, and $\alpha \x \in S$  for any $\x \in S$ and $\alpha \in \F_q$.

\begin{definition}{\em
A {\em linear $[n, k]$ code} (shortly {\em $[n, k]$ code}) $\mathcal C$ of length $n$ with dimension $k$ over $\F_q$ is a $k$-dimensional subspace of $\F_q^n$.
 If $q=2$, $\mathcal C$ is called a {\em binary} code.
}
\end{definition}

\begin{definition}{\em
A {\em generator matrix} for an $[n, k]$ code $\C$ over $\F_q$ is a
$k \times n$ matrix $G$ whose rows form a basis for $\C$.
A generator matrix of the form $[I_k | A]$ where $I_k$ is the $k \times k$ identity matrix is called in {\em standard form}.
}
\end{definition}

\begin{definition}{\em
Given an $[n, k]$ code $\C$ over $\F_q$, there is an $(n-k)\times n$ matrix $H$, called a {\em parity check matrix} for $\C$, defined
by
\[
\C = \{ \x \in  \F_q^n ~|~ H{\x}^T = {\bf 0} \}.
\]
}
\end{definition}

It is a well-known fact~\cite{HuffmanPless} that if $G = [I_k ~|~ R]$ is a generator matrix in standard form for the $[n, k]$ code $\C$ over $\F_q$,
then $H = [-R^T | I_{n-k}]$ is a parity check matrix for $\C$.

\begin{definition}{\em
The {\em syndrome} of a vector ${\bf x}$ in $\F_q^n
$ with respect to
the parity check matrix $H$ is the vector in $\F_q^{n-k}$
defined by

\[{\mbox{syn}}({\bf x}) = H{\bf x}^T.
\]
}
\end{definition}

\begin{example}
{\em
The matrix $G = [I_4 | R]$, where
\[
G =\left[\begin{array}{cccc|ccc}
1 & 0 & 0&0&0&1&1\\
0&1&0&0 &1&0&1\\
0&0&1&0& 1&1&0\\
0&0&0&1& 1&1&1\\
\end{array}
\right],
\]
is a generator matrix in standard form for a binary $[7,4]$ code $\mathcal H_3$, called the binary $[7, 4]$ Hamming code. By the above fact,
a parity check matrix for $\mathcal H_3$ is
\[
H = [R^T ~| ~I_3] =
\left[\begin{array}{cccc|ccc}
0&1&1&1& 1&0&0\\
1&0&1&1 &0&1&0 \\
1&1&0&1 &0&0&1\\
\end{array}
\right].
\]
In fact, $H$ is equivalent under elementary row operations to the below parity check matrix:
\begin{equation}\label{eq-Hamming-par}
H \sim
H_3= \left[\begin{array}{ccccccc}
0&0&0&1& 1&1&1\\
0&1&1&0 &0&1&1 \\
1&0&1&0 &1&0&1\\
\end{array}
\right].
\end{equation}
The $i$th column of $H_3$ corresponds to a binary representation of $i$ for $1 \le i \le 7$.

We describe briefly the Syndrome Decoding for $\mathcal H_3$ (see~\cite[p. 43]{HuffmanPless}).
Let ${\bf x}=(1 0 0 0 0 1 1)$ be a codeword and ${\bf y}=(1 1 0 0 0 1 1)$ be a received vector. Then $H_3{\bf x}^T=(0 0 0 0)^T$ and $H_3{\bf y}^T=(0 1 0 0)^T$. Letting ${\bf e}=(0 1 0 0 0 0 0)$, we have $H_3{\bf e}^T=(0 1 0 0)^T$. Hence, $H_3{\bf y}^T= H_3{\bf e}^T$, that is, $H_3 ({\bf y} - {\bf e})^T=(0 0 0 0)^T$. Thus, ${\bf y}$ is decoded to ${\bf x}={\bf y} -{\bf e} \in \mathcal H_3$. In fact, $\mathcal H_3$ can correct any one error.
}
\end{example}

Let $\x=(x_1, \dots, x_n), \y=(y_1, \dots, y_n) \in \F_q^n$.
The {\em Euclidean inner product} of $\x$ and $\y$ is defined by
\[
\left< \x, \y\right> =\sum_{i=1}^n x_iy_i.
\]

The {\em dual} $\C^{\perp}$ of $\C$ is defined by
\[
\C^{\perp} = \{ \x \in \F_q^n ~|~ \left<\x, {\bf c}\right> = 0 {\mbox{ for all }} {\bf c} \in \C \}.
\]
It is easy to see that if $G$ and $H$ are generator and parity check matrices, respectively, for $\C$, then $H$ and $G$ are generator and parity check matrices, respectively, for $\C^{\perp}$.

\begin{definition}{\em
An $[n,k]$ code $\C$ over $\F_q$ is {\em self-orthogonal} if $\C \subset \C^{\perp}$ and {\em self-dual} if  $\C = \C^{\perp}$.
}
\end{definition}
The length
$n$ of a self-dual code is even and its dimension is $n/2$.

\begin{definition}{\em
The{\em (Hamming) distance} $d(\x, \y)$ between two vectors $\x, \y \in  \F_q^n$ is defined to be $|\{ i ~|~ x_i \ne y_i, 1 \le i \le n \}|$. The {\em (Hamming) weight} wt($\x$)
of a vector $\x \in  \F_q^n$ is $|\{ i ~|~ x_i \ne 0, 1 \le i \le n \}|$. The {\em minimum weight} or {\em minimum distance $d$} of a linear $[n,k]$ code $\C$ is $\min\{wt(\x)~|~ {\bf 0} \ne \x \in \C \}$.

 An {\em $[n,k,d]$ code} (over $\F_q$) is an $[n, k]$ code over $\mathbb F_q$ with minimum distance $d$ (see~\cite{Hub}). For example, the binary $[7, 4]$ Hamming code $\mathcal H_3$ has minimum distance 3 and so it is a $[7,4,3]$ code.
}
\end{definition}

Given an $[n, k, d]$ code over $\F_q$, $d$ determines the error-correcting capability as follows.

\begin{theorem}{\rm (\cite[Theorem 1.11.4]{HuffmanPless})}
If $\C$ is an $[n, k, d]$ code over $\F_q$, then it can correct up to $t=\lfloor (d-1)/2 \rfloor$ errors.
\end{theorem}

\medskip

From now on, we describe basic facts on fuzzy linear codes from~\cite{ShuGan}.
Let $A$ be a fuzzy set in $\F_q^n$ with a membership function $A(\x)$.
 A fuzzy linear subspace is naturally defined as follows.

\begin{definition}[\cite{KatLiu},\cite{ShuGan}] \label{def-fuzzy linear subspace}

{\em
Let $A$ be a fuzzy set in $\F_q^n$ with a membership function $A(\x)$.
If
\begin{enumerate}

\item[{(i)}] $A(\x+ \y) \ge \min \{A(\x), A(\y) \}$ for any $\x, \y \in \F_q^n$ {\mbox{ and }}

\item[{(ii)}] $A(\lambda x ) \ge A(\x)$ for any $\x \in \F_q^n$ and $\lambda \in \F_q$, then
\end{enumerate}

we call $A$ a {\em fuzzy linear subspace} of $\F_q^n$. If $A_{\alpha} = \{ \x | A(\x) \ge \alpha \}$, it is called an {\em upper $\alpha$-level cut}.

}
\end{definition}

\begin{lem}[\cite{ShuGan}] \label{lem-lin-comb}
A fuzzy set $A$ in $\F_q^n$ is a fuzzy linear subspace if and only if for any $\alpha, \beta \in \F_q$ and $\x, \y \in \F_q^n$,
$A(\alpha \x + \beta \y) \ge \min \{A(\x), A(\y) \}$.
\end{lem}

\begin{lem}[\cite{ShuGan}]\label{lem-alpha-sub}
\label{lem-lin-space-1} A fuzzy set $A$ is a fuzzy linear subspace of $\F_q^n$ on $\F_q$ if for any $\alpha \in [0,~1]$ such that $A_{\alpha} \ne \phi,$ $A_{\alpha}$ is a linear subspace of $\F_q^n$.
\end{lem}

The converse of Lemma~\ref{lem-lin-space-1} is also true.

\begin{lem}
If $A$ is a fuzzy linear subspace of $\F_q^n$ on $\F_q$,
then for any $\alpha \in [0,~1]$ such that $A_{\alpha} \ne \phi,$ $A_{\alpha}$ is a linear subspace of $\F_q^n$.
\end{lem}

\begin{proof}
Suppose that $A$ is a fuzzy linear subspace of $\F_q^n$ on $\F_q$. Let $\alpha \in [0,~1]$ satisfy $A_{\alpha} \ne \phi$.
Assume that $\x, \y \in A_{\alpha}$. Then $A(\x) \ge \alpha$ and $A(\y) \ge \alpha$. Thus $\min \{A(\x), A(\y) \} \ge \alpha$. Therefore, as $A(\x+ \y) \ge \min \{A(\x), A(\y) \}$, we have $A(\x+ \y) \ge \alpha$, which implies $\x+ \y \in A_{\alpha}$. Similarly, if $\x \in A_{\alpha}$, then $A(\x) \ge \alpha$. Hence, since $A(\lambda \x) \ge A(\x)$ for any $\lambda \in \F_q$, we have $A(\lambda \x) \ge \alpha$. This implies that $\lambda \x \in A_{\alpha}$. Therefore, $A_{\alpha}$ is a linear subspace of $\F_q^n$.
\end{proof}

 One can naturally define a fuzzy linear code as follows.

\begin{definition}[\cite{ShuGan}] \label{def-fuzzy-lin}
{\em
A fuzzy set $A$ in $\F_q^n$ is a {\em fuzzy linear code} if for any $\alpha \in [0, ~1]$ such that $A_{\alpha} \ne \phi$, $A_{\alpha}$ is a linear code over $\F_q$.
}
\end{definition}

Let $A$ be a fuzzy linear code of length $n$ over $\mathbb F_q$.
Let Im$(A)=\{A(\x)| \x \in \F_q^n\}$. Suppose that there are $m$ elements in Im$(A)$ with an order of $\alpha_0 > \alpha_1 > \cdots > \alpha_{m-1}$.
 It is easy to see that $m \le q^n$ as there are $q^n$ elements in $\mathbb F_q^n$.
 Let the generator matrix of $A_{\alpha_k}$ in $G_{\alpha_k}$. Then $A$ can be determined by $m$ matrices
$G_{\alpha_0}, G_{\alpha_1}, \dots, G_{\alpha_{m-1}}$, where $G_{\alpha_i}$ is a submatrix of $G_{\alpha_{i+1}}$ for $i=0, \dots, m-2$. These matrices generate nested linear codes $\C_{\alpha_0} \subset \C_{\alpha_1} \subset \cdots \subset \C_{\alpha_{m-1}}$.

\begin{lem}\label{lem-one-to-one}
Let $A$ be a fuzzy linear code of length $n$ over $\mathbb F_q$.
Let Im$(A)=\{A(x)~|~ x \in \mathbb F_q^n\}$ with an order of $\alpha_0 > \alpha_1 > \cdots > \alpha_{m-1}$. Let the corresponding upper $\alpha_i$-level cut be $A_{\alpha_i}$.
If $A_{\alpha_i} = A_{\alpha_j}$, then $\alpha_i=\alpha_j$.
\end{lem}

\begin{proof}
Suppose that $\alpha_i \ne \alpha_j$. May assume that $\alpha_i > \alpha_j$. Then it is clear that
$A_{\alpha_i} \subseteq A_{\alpha_j}$. Since $\alpha_j$ is in $Im(A)$, choose $x \in A_{\alpha_j}$ such that $A(x)=\alpha_j$. Then $x$ is not in $A_{\alpha_i}$ since if so, then $A(x) \ge \alpha_i > \alpha_j=A(x)$, which is a contradiction.
Therefore, $A_{\alpha_i} \subsetneq A_{\alpha_j}$, which implies that $A_{\alpha_i} \neq A_{\alpha_j}$.
\end{proof}

\begin{example}{\em

Consider a fuzzy subset $A$ of $V_3=\mathbb F_2^3$ defined by
\[
A: V_3 \rightarrow [0,1],
\]
such that
\[A(\x) = \left\{
\begin{array}{ll}
		1  & \mbox{if } \x=(0,0,0) \\
	   1/2 & \mbox{if } \x=(1,1,0), (1,0,1), (0,1,1) \\
       1/3 & \mbox{if } \x=(1,1,1), (1,0,0), (0,1,0), (0,0,1). \\
	\end{array}
\right.
\]

It is easy to see that $A_{1}= \{(0,0,0) \}$ is a trivial linear code of length 3, $A_{1/2} = \{ (0,0,0), (1,1,0), (1,0,1), (0,1,1) \}$ is a linear $[3,2]$ code, and $A_{1/3} = V_3$. Therefore, $A$ is a fuzzy linear code in $\mathbb F_2^3$.
}
\end{example}

\section{Arithmetic of fuzzy linear codes and extension principle}

As in~\cite{SyrGra}, we can consider the arithmetic of fuzzy linear codes.

\begin{definition} \label{def-sum} Let  $A$ and $B$ be fuzzy linear codes in $\F_q^n$. Define
\begin{enumerate}

\item [{\rm {(i)}}] $(A \cap B)(\x)= \min \{A(\x), {\color{blue}{B(\x)}} \}$ for any $\x \in \F_q^n$,

\item [{\rm {(ii)}}] $(A \cup B)(\x) = \max \{A(\x), {\color{blue}{B(\x)}} \}$ for any $\x \in \F_q^n$.

\end{enumerate}

\end{definition}

\begin{proposition}
Let  $A$ and $B$ be fuzzy linear codes in $\F_q^n$. Then $A \cap B$ is a fuzzy linear code
in $\F_q^n$.
\end{proposition}

\begin{proof}
Note that $A \cap B$  is a fuzzy set.
We want to show that $A \cap B$ is a fuzzy linear code in $\F_q^n$.

Let $\x, \y \in \mathbb F_q^n$.
Since $A$ and $B$ are fuzzy linear codes in $\F_q^n$, we have
$A(\x+ \y) \ge \min\{ A(\x), A(\y)$ and $B(\x+ \y) \ge \min\{ B(\x), B(\y)\}$.
Therefore, we obtain the following:
\begin{eqnarray*}
(A \cap B)(\x+\y)& =    &\min \{ A (\x+ \y), B(\x+ \y) \} {\mbox{ by Definition~\ref{def-sum} }} \\
                 & \ge  &\min \{ \min\{ A(\x), A(\y)\},   \min\{ B(\x), B(\y)\} \} \\
                 & =     &\min\{ A(\x), A(\y),  B(\x), B(\y) \} \\
                 & =     &\min\{ A(\x), B(\x),  A(\y), B(\y) \} \\
                 & =     &\min \{ \min\{ A(\x), B(\x) \},  \min\{ A(\y), B(\y)\} \} \\
                 & =     &\min \{ (A \cap B)(\x),  (A \cap B)(\y) \} {\mbox{ by  Definition~\ref{def-sum}. }}
\end{eqnarray*}
Hence,
\[(A \cap B)(\x+\y) \ge \min\{(A \cap B)(\x), (A \cap B)(\y) \} {\mbox{ for any }} \x, \y \in \mathbb F_q^n.
\]

Furthermore, since $A$ and $B$ are fuzzy linear codes, we have
  $A (\lambda \x) \ge A(x)$ and $B(\lambda \x) \ge B(\x)$, for any $\x \in \mathbb F_q^n$ and $\lambda \in \F_q$. Therefore, we have the following:
\begin{eqnarray*}
(A \cap B)(\lambda \x) & =    &  \min \{ A (\lambda \x), B(\lambda \x) \} \\
                       & \ge  &  \min \{ A(\x),   B(\x)  \} \\
                       & =  &  (A \cap B)(\x).
\end{eqnarray*}

By Definition~\ref{def-fuzzy linear subspace}, $A \cap B$ is a fuzzy linear code in $\mathbb F_q^n$.
\end{proof}

\begin{remark}{\em
It is easy to see that $A \cup B$ is not a fuzzy linear code in $\mathbb F_q^n$. For example,
consider a fuzzy linear code $A$ in $\mathbb F_2^2$ defined by
\[
A: \mathbb F_2^2 \rightarrow [0,1] {\mbox{ such that }} A(\x) = \left\{
\begin{array}{ll}
		1  & \mbox{if } \x \in \left< (1,0) \right> \\
	    0  & \mbox{otherwise} \\
	\end{array}
\right.
\]
and a fuzzy linear code $B$ in $\mathbb F_2^2$ defined by
\[
B: \mathbb F_2^2 \rightarrow [0,1] {\mbox{ such that }} B(\x) = \left\{
\begin{array}{ll}
		1  & \mbox{if } \x \in \left< (0,1) \right> \\
	    0  & \mbox{otherwise.} \\
	\end{array}
\right.
\]

Let ${\bf e}_1 = (1,0)$ and ${\bf e}_2 =(0,1)$, so that ${\bf e}_1 + {\bf e}_2=(1,1)$.
Then
\[
(A \cup B)({\bf e}_1 + {\bf e}_2)= \max \{ A({\bf e}_1 + {\bf e}_2), B({\bf e}_1 + {\bf e}_2) \} =\max\{ 0, 0 \}=0.
\]

On the other hand,
\begin{eqnarray*}
& &\min\{(A \cup B)({\bf e}_1), (A \cup B)({\bf e}_2) \} \\
& = &\min \{ \max \{ A({\bf e}_1), B({\bf e}_1) \},  \max \{ A({\bf e}_2), B({\bf e}_2) \} \}\\
& =&  \min \{ 1, 1\} =1.
\end{eqnarray*}
Hence,
$(A \cup B)({\bf e}_1 + {\bf e}_2) \not \ge  \min\{(A \cup B)({\bf e}_1), (A \cup B)({\bf e}_2) \}$.
 This implies that $A \cup B$ is not a fuzzy linear code in $\mathbb F_2^2$.
}
\end{remark}

Using the extension principle (see~\cite[Sec. 3.2.3]{SyrGra}), we can define the sum of two fuzzy linear codes as follows.
\begin{definition} \label{def-fuzzy-sum}
Let  $A$ and $B$ be fuzzy linear codes in $\F_q^n$. Let $\z \in \F_q^n$. Define a fuzzy set $A+B$ in $\F_q^n$ as follows:
\begin{eqnarray*}
(A + B)(\z) & = & \max_{\z=\x+\y} \min \{ A(\x), B(\y) \} \\
            & = & {\color{blue}{\max_{\x} ~ \min \{ A(\x), B(\z -\x)}} \}   
\end{eqnarray*}
\end{definition}

\begin{proposition} Let $A$ and $B$ be fuzzy linear codes in $\F_q^n$. {\color{blue}{Then
the sum $A + B$ is a fuzzy linear code in $\F_q^n$.}}
\end{proposition}

\begin{proof}

{\color{blue}{
The following reference (a) proved this proposition.

\begin{enumerate}
\item[{(a)}]  P. Lubczonok, Fuzzy vector spaces, Fuzzy Sets Syst. 1990, 38, 329-343.
\end{enumerate}
However, the proof of Proposition 4.12 in Reference (a) is not complete since it only considers Condition (i) of Definition~\ref{def-fuzzy linear subspace}. Hence we include the proof for completeness. Note that the vector space in Reference (a) seems to assume the real vector space but the argument also works for any finite field.

To get a contradiction, suppose that Condition (i) of Definition~\ref{def-fuzzy linear subspace} is not satisfied. Then we have
\[
(*)~ (A + B)(\z_1 + \z_2) <  \min \{(A + B)(\z_1), (A + B)(\z_2) \} {\mbox{ for some }} \z_1, \z_2 \in \F_q^n.
\]
By Definition~\ref{def-fuzzy-sum}, we have that $(A + B)(\z_1) = \min \{ A(\x_1), B(\z_1 - \x_1) \}$ for some $\x_1 \in \F_q^n$ and that $(A + B)(\z_2) = \min \{ A(\x_2), B(\z_2 - \x_2) \}$ for some $\x_2 \in \F_q^n$.
Hence the right hand side of $(*)$ is computed as follows.
\begin{eqnarray*}
&   & \min \{(A + B)(\z_1), (A + B)(\z_2) \} \\
& = & \min \{\min \{ A(\x_1), B(\z_1 - \x_1) \}, \min \{ A(\x_2), B(\z_2 - \x_2) \}  \}   \\
& = & \min \{ A(\x_1), A(\x_2), B(\z_1 - \x_1), B(\z_2 - \x_2) \} \\
& \leq  & \min \{ A(\x_1 + \x_2), B(\z_1 +\z_2 - \x_1 -\x_2) \} {\mbox{ since }} A, B {\mbox{ are fuzzy linear}}\\
& \leq  & (A + B)(\z_1 + \z_2) {\mbox{ by Definition~\ref{def-fuzzy-sum}.}}
\end{eqnarray*}
Then combining $(*)$ together, we obtain
\[
(A + B)(\z_1 + \z_2) < (A + B)(\z_1 + \z_2).
\]
This is a contradiction. Therefore, Condition (i) of Definition~\ref{def-fuzzy linear subspace} is satisfied. 

Next we prove Condition (ii) of Definition~\ref{def-fuzzy linear subspace}.
Let $\z \in \F_q^n$ and $\lambda \in \F_q$.
Then 
\begin{eqnarray*}
 (A+B)(\lambda z) & = & \max_{\lambda\z=\x+\y} \min \{ A(\x), B(\y) \} \\ 
                   & = & \max_{\z=\lambda^{-1} \x+ \lambda^{-1}\y} \min \{ A(\x), B(\y) \} \\
                   & = & \max_{\z=\lambda^{-1} \x+ \lambda^{-1}\y} \min \{ A(\lambda \lambda^{-1}\x), B(\lambda \lambda^{-1} \y) \}  \\
                   & \ge & \max_{\z=\lambda^{-1} \x+ \lambda^{-1}\y} \min \{ A(\lambda^{-1}\x), B(\lambda^{-1} \y) \}  \\
                  & = & \max_{\z=\x_1+ \y_1} \min \{ A(\x_1), B(\y_1) \}  \\
                  & = & (A+B)(\z).
\end{eqnarray*}
Therefore, $A+B$ is a fuzzy linear code in $\F_q^n$.
}}
\end{proof}

{\color{blue}{We describe this proposition using an example. The errors in this example from the original version were pointed out by Hidayat Budiawan on May 9, 2024.}}

Consider fuzzy linear codes $A$ and $B$ in $\F_2^3$ as follows:
\[
A: \mathbb F_2^3 \rightarrow [0,1] {\mbox{ such that }} A(\x) = \left\{
\begin{array}{ll}
		1   & \mbox{if } \x=(000) \\
	    1/2 & \mbox{if } \x=(110), (101), (011) \\
        1/3 & \mbox{if } \x=(111), (100), (010), (001) \\
	\end{array}
\right.
\]
and
\[
B: \mathbb F_2^3 \rightarrow [0,1] {\mbox{ such that }} B(\x) = \left\{
\begin{array}{ll}
		 1 & \mbox{if } \x=(000), (111) \\
	     0 & \mbox{otherwise.}
	\end{array}
\right.
\]
It is easy to see that the upper $\alpha$-level cuts $A_1, A_{1/2}, A_{1/3}, B_{1/4}, B_0$ are linear codes.
Then using the definition of $A+B$,
\begin{eqnarray*}
(A+B)(000) & = & \max_{(000)=\x+\y} \min \{ A(\x), B(\y) \} \\
             & = & \max_{(000)=\x+\y} \{ \min \{ A(000), B(000) \}, \min \{ A(111), B(111)\},\\
             &  &   ~~~~~~~~~~~~~  \min \{ A(110), B(110) \}, \cdots  \} \\
             & = & \max_{(000)=\x+\y} \{ 1, 1, 0, \cdots  \} \\
             & = & 1.
\end{eqnarray*}

Similarly, $(A+B)(111) = 1$. In fact, we have
$
A+B: \mathbb F_2^3 \rightarrow [0,1]
{\mbox{ such that }} $

\[(A+B)(\x) = \left\{
\begin{array}{ll}
		1 & \mbox{if } \x=(000), (111) \\
      1/2 & \mbox{if } {\color{blue}{ \x=(110), (101), (011), (100), (010), (001).}}\\
	\end{array}
\right.
\]
{\color{blue}{
Since the upper $\alpha$-level cuts $(A+B)_{1}$ and $(A+B)_{1/2} = \mathbb F_2^3$
are linear codes, $A+B$ is a fuzzy linear code in $\F_2^3$ by Definition~\ref{def-fuzzy-lin}.}}

\medskip

{\color{blue}{We also see that the direct sum $A \oplus B$ is a fuzzy linear code as follows.}}

\begin{definition} \label{def-direct-sum}
{\em Let  $A$ and $B$ be fuzzy linear codes in $\F_q^n$ and $\F_q^m$, respectively.
The {\em direct sum} of $A$ and $B$ is defined as the fuzzy set $A \oplus B$ in $\F_q^n \oplus \F_q^m$
satisfying
 $(A \oplus B)(\z)= \min \{A(\x), {\color{blue}{B}}(\y) \}$ for any $\z=\x + \y$ where $\x \in \F_q^n$ and $\y \in  \F_q^m$.
}
\end{definition}

\begin{proposition}
Let  $A$ and $B$ be fuzzy linear codes in $\F_q^n$ and $\F_q^m$, respectively. Then $A \oplus B$ is a fuzzy linear code in $\F_q^n \oplus \F_q^m$.
\end{proposition}

\begin{proof}
Let $\z_1=\x_1 + \y_1$ and $\z_2=\x_2 + \y_2$, where $\x_1, \x_2 \in \F_q^n$ and $\y_1, \y_2 \in \F_q^m$.
Then since $A$ and $B$ are fuzzy linear codes, we have
$A(\x_1 + \x_2) \ge \min \{ A(\x_1), A(\x_2) \}$ and  $B(\y_1 + \y_2) \ge \min \{ B(\y_1), B(\y_2) \}$.

Therefore, we have the following:
\begin{eqnarray*}
  (A\oplus B)(\z_1 + \z_2) &=& (A\oplus B)((\x_1 + \x_2) + (\y_1 + \y_2)) \\
   &=& \min\{ A(\x_1 + \x_2), B(\y_1 + \y_2)  \} {\mbox{ by Definition \ref{def-direct-sum}}}\\
   &\ge& \min\{ \min \{ A(\x_1), A(\x_2) \},  \min \{ B(\y_1), B(\y_2) \} \} \\
      & = & \min\{ A(\x_1), A(\x_2), B(\y_1), B(\y_2) \} \\
        & = & \min\{ A(\x_1), B(\y_1), A(\x_2), B(\y_2) \} \\
   & = & \min\{ \min \{ A(\x_1), B(\y_1) \},  \min \{ A(\x_2), B(\x_2) \} \} \\
      & = & \min\{ (A \oplus B)(\x_1+ \y_1),  (A \oplus B)(\x_2+ \y_2) \} \\
   & = & \min\{ (A \oplus B)(\z_1),  (A \oplus B)(\z_2) \}.
\end{eqnarray*}
Similarly, since $A$ and $B$ are fuzzy linear codes, we have $A(\lambda \x_1) \ge   A(\x_1)$ and $B(\lambda \y_1) \ge  B(\y_1)$ for any $\lambda \in \F_q$. Thus, we have the following:
\begin{eqnarray*}
(A\oplus B)(\lambda \z_1) & = & (A\oplus B)(\lambda\x_1 + \lambda \y_1) \\
                        & = & \min \{ A(\lambda\x_1), B(\lambda \y_1) \} {\mbox{ by Definition \ref{def-direct-sum}}}\\
                        & \ge & \min \{  A(\x_1), B(\y_1) \} \\
                        & = &  (A\oplus B)(\x_1+ \y_1) \\
                        & = & (A\oplus B)(\z_1).
\end{eqnarray*}
This completes the proof.
\end{proof}

\section{Fuzzy self-dual codes}

Note that
a given linear code $\C$ defines a fuzzy linear code $A$  where a sequence of linear codes $\A_{\alpha_0} \subset \A_{\alpha_1} \subset \cdots \subset \A_{\alpha_i} =\C  \subset  \cdots \subset \A_{\alpha_{m-1}}$ for some $i$ and $m$ is assigned.  One problem of this association is that there are many non-canonical choices of supercodes $\A_{\alpha_{j'}}$ of $\C$ ($i< j' \le m-1$) of $\C$. If we choose a self-orthogonal code $\C$, then subcodes $\A_{\alpha_j}$ of $\C$ ($j < i$) determine supercodes $\A_{\alpha_{j'}}$ of $\C$ ($j' > i$) because $\A_{\alpha_j}^{\perp} =\A_{\alpha_{j'}}$. From now on, we consider self-orthogonal or self-dual codes and associate them to fuzzy codes, called fuzzy self-orthogonal or self-dual codes.

We recall that the dual of a fuzzy linear code over the modular ring $\mathbb Z_{p^k}$ was defined in~\cite{Ata} using the module theory. This can be naturally defined over $\mathbb F_q$ as follows.

\begin{definition}[\cite{Ata}]
{\em
Let $A$ and $B$ be two fuzzy linear codes over $\mathbb F_q$ . We say that $A$ is {\em orthogonal} to $B$, denoted by $A \perp B$, if $Im(B) = \{1-c ~|~ c \in Im(A) \}$ and for all $t \in [0, 1], B_{1-t} = (A_t)^{\perp} = \{\y \in  \mathbb F_q^n   | \left<\x, \y\right> = 0 {\mbox{ for all }} \x \in A_t \}$.
}
\end{definition}

\begin{theorem}
Let $A$ be a fuzzy linear code over $\mathbb F_q$. Then there exists a
fuzzy linear code $B$ over $\mathbb F_q$ such that $A \perp B$ if and only if $|Im(A)| > 1$  and for any $\gamma \in Im(A)$ there exists $\epsilon \in Im(A)$ such that $A_{\gamma} = (A_{\epsilon})^{\perp}$.
\end{theorem}

\begin{proof}
The proof is basically the same as that of~\cite[Theorem 3.15]{Ata}.
\end{proof}

\begin{theorem}
If $B$ and $C$ are fuzzy linear codes over $\mathbb F_q$ such that $B$ is orthogonal to $A$ and $C$ is orthogonal to $A$, then $B=C$.
\end{theorem}

\begin{proof}
The proof is basically the same as that of~\cite[Theorem 3.16]{Ata}.
\end{proof}

Therefore, if there exists a fuzzy linear code $B$ orthogonal to $A$, then $B$ is unique.
The code $B$ is called the {\em dual} of a fuzzy linear code $A$ and is denoted by $A^{\perp}$.

\begin{cor}
Let $A$ be a fuzzy linear code over $\mathbb F_q$. If $A^{\perp}$ exists, then $(A^{\perp})^{\perp} = A$.
\end{cor}



There is no literature on the definition of a fuzzy self-dual code. Now we can define a fuzzy self-dual code and a fuzzy self-orthogonal code for the first time.

\begin{definition}\label{def-sd-so}
{\em
Let $A$ be a fuzzy linear code over $\mathbb F_q$.

\begin{enumerate}

\item[(i)]
 If $|Im(A)| > 1$  and $(A_{1-\alpha})^{\perp}=A_{\alpha} $ (equivalently, $A_{\alpha}^{\perp} = A_{1-\alpha}$) for any $\alpha \in Im(A)$ and $A_{\beta}^{\perp}=A_{\beta}$ for some $\beta$, then $A$ is called a {\em fuzzy self-dual code}.

\item[(ii)]
If $|Im(A)| > 1$  and $(A_{1-\alpha})^{\perp}=A_{\alpha} $ (equivalently, $A_{\alpha}^{\perp} = A_{1-\alpha}$), then $A$ is called a {\em fuzzy self-orthogonal code}.

\end{enumerate}
}
\end{definition}

\begin{lem}
If $A$ is a fuzzy self-dual code over $\mathbb F_q$ such that $A_{\beta}$ is a self-dual code for some $\beta$, then
$\beta=\frac{1}{2}$.
\end{lem}

\begin{proof}
By hypothesis, we have $A_{\beta}^{\perp}= A_{\beta}$.
By Definition~\ref{def-sd-so} (i), we know $|Im(A)| > 1$ and $A_{\alpha} = (A_{1-\alpha})^{\perp}$ for any $\alpha \in Im(A)$.
Letting $\alpha=\beta$, we get $A_{\beta}^{\perp}= A_{\beta}=(A_{1-\beta})^{\perp}$, that is, $A_{\beta}=A_{1-\beta}$. Therefore, by Lemma~\ref{lem-one-to-one} we have $\beta=1-\beta$, that is, $\beta=\frac{1}{2}$ as wanted.
\end{proof}

\begin{example}
{\em
Consider a fuzzy subset $A$ of $V_4=\mathbb F_2^4$:

\[
A: V_4 \rightarrow [0,1],
\]
such that
\[A(\x) = \left\{
\begin{array}{ll}
		1  & \mbox{ if } \x=(0,0,0,0) \\
	   1/2 & \mbox{ if } \x=(1,0,1,0), (0,1,0,1), (1,1,1,1) \\
       0 & \mbox{ for other } \x \in V_4. \\
	\end{array}
\right.
\]

It is easy to check that
$A_{1}= \{(0,0,0,0) \}$ is a trivial linear code of length 4,
$A_{1/2} = \{ (0,0,0,0), (1,0,1,0), (0,1,0,1), (1,1,1,1) \} =A_{(1-1/2)}$ is a linear $[4,2]$ code, and $A_{0} = V_4=A_{(1-0)}^{\perp}$ is a full space.
Hence, $A$ is a fuzzy self-dual code.
}
\end{example}






Let $\overline{\mathcal H}_3$ be the extended Hamming $[8,4,4]$ code in $V_8=\mathbb F_2^8$, which is self-dual, that is, $\overline{\mathcal H}_3=\overline{\mathcal H}_3^{\perp}$. It has the following generator matrix.
\begin{equation}\label{eq-Ham-8-4-4}
G(\overline{\mathcal H}_3)=\left[\begin{array}{cccccccc}
1 & 0 & 0 & 0& 0& 1& 1 & 1 \\
0 & 1 & 0 & 0& 1& 0& 1 & 1 \\
0 & 0 & 1 & 0& 1& 1& 0 & 1 \\
0 & 0 & 0 & 1& 1& 1& 1 & 0 \\
\end{array}
\right]
=\left[\begin{array}{c}
{\bf r}_1 \\
{\bf r}_2 \\
{\bf r}_3 \\
{\bf r}_4 \\
\end{array}
\right]
\end{equation}

We introduce two ways to make $\overline{\mathcal H}_3$ as a fuzzy self-dual code.

\begin{example}
{\em

We define a fuzzy self-dual code $B$ associated with $\overline{\mathcal H}_3$. Let

\[
B: V_8 \rightarrow [0,1],
\]

such that
\[B(\x) = \left\{
\begin{array}{ll}
		1  & \mbox{ if } x=0 \\
	   1/2 & \mbox{ if } x \in \overline{\mathcal H}_3 -\{ 0 \} \\
       0 & \mbox{ if } x \in V_8 - \overline{\mathcal H}_3 \\
	\end{array}
\right.
\]

One can check that
$B_{1}= \{ 0 \}$, $B_{(1-1)}^{\perp}=V_8^{\perp}=\{0\}$,
and
$B_{1/2} = \overline{\mathcal H}_3=B_{(1-1/2)}^{\perp}$.
Hence, $B$ is a fuzzy self-dual code.
}
\end{example}

\begin{example}
\label{eg-code-D}
{\em

 Now we can also associate another fuzzy self-dual code $D$ with $\overline{\mathcal H}_3$ having $G(\overline{\mathcal H}_3)$ as a generator matrix. Let $S_0=\{0 \}=V_8^{\perp}, S_1 =\left <{\bf r}_1 \right>, S_2=\left <{\bf r}_1, {\bf r}_2 \right>, S_3=\left <{\bf r}_1, {\bf r}_2, {\bf r}_3 \right>, S_4=\overline{\mathcal H}_3$, where ${\bf r}_i$s ($1 \le i \le 3$) are as in Equation (\ref{eq-Ham-8-4-4}).

 Let

\[
D: V_8 \rightarrow [0,1],
\]

such that
\[D(\x) = \left\{
\begin{array}{ll}
		1  & \mbox{ if } \x \in S_0 \\
	   7/8 & \mbox{ if } \x \in S_1 - S_0 \\
	   6/8 & \mbox{ if } \x \in S_2 - S_1  \\
	   5/8 & \mbox{ if } \x \in S_3 - S_2  \\
	   1/2 & \mbox{ if } \x \in S_4 - S_3  \\
	   3/8 & \mbox{ if } \x \in S_3^{\perp} - S_4^{\perp}  \\
	   2/8 & \mbox{ if } \x \in S_2^{\perp} - S_3^{\perp}  \\
	   1/8 & \mbox{ if } \x \in S_1^{\perp} - S_2^{\perp}  \\
	   0 & \mbox{ if } \x \in V_8 - S_1^{\perp}.  \\
	\end{array}
\right.
\]

One can check that
$D_{1 - \frac{i}{8}}=S_i$ for $0 \le i \le 4$ and
$D_{\frac{i}{8}}=S_i^{\perp}$ for $1 \le i \le 4$. Thus
$D_{1 - \frac{i}{8}}^{\perp} = D_{\frac{i}{8}}$ for $0 \le i \le 8$.
Hence, $D$ is a fuzzy self-dual code.
Furthermore, we have partitioned $V_8$ into the nine disjoint subsets of the forms $S_0, S_{i} -S_{i-1}$, and $S_{i-1}^{\perp} -S_i^{\perp}$ for $1 \le i \le 4$.
 Note also that
$|S_{i} -S_{i-1}|=2^{i-1}$ and
$|S_{i-1}^{\perp} -S_i^{\perp}|=2^{i-1}$ for $1 \le j \le 4$.
}
\end{example}

In what follows, we give a most general form of a fuzzy self-dual code over $\mathbb F_q$.

\begin{theorem}\label{thm-fuzzy-sd}
Let $\mathcal C$ be a linear self-dual $[2n, n]$ code over $\mathbb F_q$ with a basis $\{ {\bf s}_1, \dots, {\bf s}_n \}$. Let $S_0={{\bf 0}}$ and $S_i =\left< {\bf s}_1, \dots, {\bf s}_i \right>$ for $1 \le i \le n$. Let $\alpha_0 \ge \alpha_1 \ge \alpha_2 \ge \cdots \ge \alpha_{n-1} > \alpha_n=1/2$ be numbers in the interval $[0,1]$, where $\alpha_i$'s can be repeated for generality. Let $A: V_{2n} \rightarrow [0,1]$ be its associated fuzzy set defined by
\[
A(\x) = \left\{
\begin{array}{ll}
		\alpha_0  & \mbox{ if } \x \in S_0 \\
	   \alpha_{i} & \mbox{ if } \x \in S_{i} - S_{i-1} \mbox{ for } 1 \le i \le n-1 \\

	   \alpha_n=1/2  & \mbox{ if } \x \in  S_{n} - S_{n-1}=\mathcal C -S_{n-1}\\
	   1-\alpha_{i} & \mbox{ if } \x \in S_{i-1}^{\perp} - S_{i}^{\perp}  \mbox{ for } 1 \le i \le n-1  \\
	   1-\alpha_0 & \mbox{ if } \x \in V_{2n} - S_0^{\perp}.  \\
	\end{array}
\right.
\]
Then $A$ is a fuzzy self-dual code over $\mathbb F_q$.
\end{theorem}
\begin{proof}
Since $Im(A) \supset \{\alpha_0, 1/2, 1-\alpha_0 \}$, we have $|Im(A)| > 1$. It suffices to show $A_{\alpha} = (A_{1-\alpha})^{\perp}$ (equivalently, $A_{1-\alpha}=A_{\alpha}^{\perp}$) for any $\alpha \in Im(A)$. By symmetry, it is enough to consider $\alpha_i$'s for $1 \le i \le n$.
By the definition of $A(\x)$, we have $A_{1-\alpha_i}=S_{i-1}^{\perp}=A_{\alpha_i}^{\perp}$ for $1 \le i \le n$. Clearly $A_{1-\alpha_0}=V_{2n}={\bf 0}^{\perp}=A_{\alpha_0}^{\perp}$. We also have $\mathcal C = A_{1/2}=A_{1/2}^{\perp}$. Therefore, $A$ is a fuzzy self-dual code over $\mathbb F_q$.
\end{proof}

\begin{cor}
Let $\mathcal C$ be a self-dual $[2n, n]$ code over $\mathbb F_q$. Suppose that $A$ is a fuzzy self-dual code such that $A_{1/2}= \mathcal C$ and
\[
Im(A)=\{\alpha_{i_1}, \dots, \alpha_{i_k}, 1/2, \alpha_{1-i_k}, \dots, \alpha_{1-i_1} \}
\]
 for some $k>1$, where $\alpha_{i_1} > \alpha_{i_2} > \cdots > \alpha_{i_k} > 1/2 > \alpha_{1-i_k} > \dots > \alpha_{1-i_1}$. Then $A$ is obtained by function $A(\x)$ in Theorem~\ref{thm-fuzzy-sd}.
\end{cor}
\begin{proof}
Since $A$ is a fuzzy self-dual code, we have $A_{1-\alpha_{i_j}}=A_{\alpha_{i_j}}^{\perp}$ for $1 \le j \le k$. Hence by the order of $\alpha_i$'s, we obtain
$A_{\alpha_{i_1}} \subsetneq A_{\alpha_{i_2}} \subsetneq \cdots \subsetneq A_{\alpha_{i_k}} \subsetneq \mathcal C \subsetneq
A_{\alpha_{i_k}}^{\perp} \subsetneq \cdots \subsetneq A_{\alpha_{i_2}}^{\perp} \subsetneq A_{\alpha_{i_1}}^{\perp}$.
Therefore all the linear codes arising from this sequence are determined by the subcodes of $\mathcal C$ and their duals.
Choose a basis $S_{i_1}$ for $A_{\alpha_{i_1}}$ and extend it to a basis $S_{i_2}$ for $A_{\alpha_{i_2}}$ etc to get a basis for $C$. Then apply the function $f_A(x)$ in Theorem~\ref{thm-fuzzy-sd}.
\end{proof}

\begin{definition}\label{def-special-fuzzy-sd}
{\em
A {\em dimension-parameterized fuzzy self-dual code} is a fuzzy self-dual $[2n, n]$ code $A$ over $\mathbb F_q$ with $\alpha_i=1- \frac{{\mbox{dim}}S_i}{2n}$ for $0 \le i \le 2n$ defined by Theorem~\ref{thm-fuzzy-sd}, where $S_i$ is an $i$-dimensional subcode of the self-dual code $A_{1/2}$.
 }
\end{definition}

\begin{example}

{\em
Let $\mathcal G_{24}$ be the famous binary extended Golay $[24,12,8]$ code~\cite[Section 1.9]{HuffmanPless}. It has generator matrix $G_{24}=[I_{12} ~|~ R_{12}]$ in standard form, where
\[
R_{12}= \left[\begin{array}{c|ccccccccccc}
0 & 1 &1 &1 &1 &1 &1 &1 &1 &1 &1 &1 \\
\hline
1 &1 &1 &0 &1 &1 &1 &0 &0 &0 &1 &0 \\
1 &1 &0 &1 &1 &1 &0 &0 &0 &1 &0 &1 \\
1 &0 &1 &1 &1 &0 &0 &0 &1 &0 &1 &1 \\
1 &1 &1 &1 &0 &0 &0 &1 &0 &1 &1 &0 \\
1 &1 &1 &0 &0 &0 &1 &0 &1 &1 &0 &1 \\
1 &1 &0 &0 &0 &1 &0 &1 &1 &0 &1 &1 \\
1 &0 &0 &0 &1 &0 &1 &1 &0 &1 &1 &1 \\
1 &0 &0 &1 &0 &1 &1 &0 &1 &1 &1 &0 \\
1 &0 &1 &0 &1 &1 &0 &1 &1 &1 &0 &0\\
1 &1 &0 &1 &1 &0 &1 &1 &1 &0 &0 &0\\
1 &0 &1 &1 &0 &1 &1 &1 &0 &0 &0 &1\\
\end{array}
\right].
\]

 Then by Definition~\ref{def-special-fuzzy-sd}, let
$\alpha_i=1- \frac{i}{24}$ for $0 \le i \le 24$.
For $1 \le i \le 12$, let $S_i$ be the subcode generated by the first $i$ rows of $G_{24}$ and set $A_{\alpha_i}=S_i$.
Then we obtain a dimension-parameterized fuzzy self-dual Golay code.
}

\end{example}

\section{Fuzzy simplex and Hamming codes}

In this section, we introduce interesting fuzzy self-orthogonal codes. We recall that the $[7,3,4]$ simplex code $\mathcal S_3$ is a well known binary self-orthogonal code whose generator matrix is below.
\[
G(\mathcal S_3)=\left[\begin{array}{ccccccc}
0 & 0 & 0 & 1& 1& 1& 1 \\
0 & 1 & 1 & 0& 0& 1& 1 \\
1 & 0 & 1 & 0& 1& 0& 1 \\
\end{array}
\right]
=\left[\begin{array}{c}
{\bf s}_1 \\
{\bf s}_2 \\
{\bf s}_3 \\
\end{array}
\right]
\]

 We can make it as an interesting fuzzy self-orthogonal code as follows.

\begin{example}\label{eg-simplex}
{\em
(Fuzzy simplex $[7,3,4]$ code and Hamming $[7,4,3]$ code)

 We construct a fuzzy self-orthogonal code $E$ with the simplex code $\mathcal S_3$ having generator matrix $G(\mathcal S_3)$. Note that $\mathcal S_3^{\perp}$ is the Hamming $[7,4,3]$ code $\mathcal H_3$.
 Let $S_0=\{0 \}=V_7^{\perp},~ S_1 =\left <{\bf s}_1 \right>,~ S_2=\left <{\bf s}_1,~ {\bf s}_2 \right>,~ S_3=\left <{\bf s}_1, {\bf s}_2, {\bf s}_3 \right>=\mathcal S_3$ such that $S_3 \subsetneq S_3^{\perp}=\mathcal H_3$. Hence we have $S_1 \subsetneq S_2 \subsetneq S_3 \subsetneq S_3^{\perp} \subsetneq S_2^{\perp} \subsetneq S_1^{\perp} \subsetneq S_0^{\perp}=V_7$.

 Let

\[
E: V_7 \rightarrow [0,1],
\]

such that
\[E(\x) = \left\{
\begin{array}{ll}
		1  & \mbox{ if } \x \in S_0 \\
	   6/7 & \mbox{ if } \x \in S_1 - S_0 \\
	   5/7 & \mbox{ if } \x \in S_2 - S_1  \\
	   4/7 & \mbox{ if } \x \in S_3 - S_2  \\
	   3/7 & \mbox{ if }  \x \in S_3^{\perp} - S_3 =\mathcal H_3 - S_3  \\
	   2/7 & \mbox{ if }  \x \in S_2^{\perp} - S_3^{\perp}  \\
	   1/7 & \mbox{ if } \x \in S_1^{\perp} - S_2^{\perp}  \\
	   0 & \mbox{ if } \x \in V_7 - S_1^{\perp}.  \\
	\end{array}
\right.
\]

One can check that
$E_{1 - \frac{i}{7}}=S_i$ for $0 \le i \le 3$ and
$E_{\frac{i}{7}}=S_i^{\perp}$ for $0 \le i \le 3$. Therefore, we have $E_{1 - \frac{i}{7}}^{\perp} = E_{\frac{i}{7}}$ for $0 \le i \le 3$. By Definition~\ref{def-sd-so}, this implies that $E$ is a fuzzy self-orthogonal code.

Furthermore, we have partitioned $V_7$ into the eight disjoint subsets of the forms $S_0, S_{i} -S_{i-1}$ for $1 \le i \le 3$, $\mathcal H_3 - S_3$, $S_{i}^{\perp} -S_{i+1}^{\perp}$ for $0 \le i \le 2$ since $S_0^{\perp} =V_7$.

Note that although $\mathcal H_3$ is not self-orthogonal, its associated fuzzy linear code is  fuzzy self-orthogonal because the concept of a fuzzy self-orthogonal code involves a sequence of pairs of a linear code and its dual.
}
\end{example}

\begin{example}{\em
(Embedding a self-orthogonal code)

In general, let $\mathcal C$ be any $[n, k]$ self-orthogonal code over $\mathbb F_q$. Then we have $S_0=\{{\bf 0} \} \subset \mathcal C \subset \mathcal C^{\perp} \subset V_n$.
Define
\[
F: V_n \rightarrow [0,1],
\]

such that
\[F(\x) = \left\{
\begin{array}{ll}
		1  & \mbox{ if } \x \in S_0 \\
	   1- k/n & \mbox{ if } \x \in \mathcal C - S_0 \\
	    k/n & \mbox{ if } \x \in \mathcal C^{\perp} - \mathcal C  \\
	   0 & \mbox{ if } \x \in V_n - \mathcal C^{\perp}.  \\
	\end{array}
\right.
\]

It is easy to check that $F_{1-k/n}^{\perp} = F_{k/n}$ and $F_1^{\perp} = F_0$. Therefore, $F$ is a fuzzy self-orthogonal code.
}
\end{example}

\section{Fuzzy Reed-Muller codes}\label{sec-RM}
In this section, we show that the family of binary nested Reed-Muller codes $\mathcal R(r,m)$ can be parameterized as upper $\alpha$-level cuts, which produce an interesting fuzzy self-orthogonal code.

The binary Reed-Muller codes  and their majority logic decoding algorithm were constructed in 1954 by Reed and Muller. Since then, Reed-Muller codes became important practically. They also have close connections with Boolean functions~\cite{MS} and finite affine and projective geometries~\cite{AssKey}. Therefore, it is interesting to study fuzzy Reed-Miller codes in terms of a fuzzy linear code.

We begin with a definition of binary Reed-Muller codes~\cite[p. 34]{HuffmanPless} as follows.

 \begin{definition}{\em
 Let $m$ be a positive integer and $r$ a nonnegative integer with $r \le m$.
 The {\em binary $r$th order Reed-Muller, or RM, code of length $2^m$,} denoted by $\mathcal R(r, m)$, has a recursive generator matrix $G(r,m)$:
   \[
G(r,m)=\left[ \begin{array}{cc}
G(r, m-1) & G(r, m-1) \\
O  & G(r-1, m-1) \\
   \end{array}
   \right],
   \]
   where $G(0,m)=[1 1 \cdots 1]$ and $G(m,m)=I_{2^m}$.
   The code $\mathcal R(0, m)$ is the binary repetition code
of length $2^m$ with generator matrix $G(0,m)$ and $\mathcal R(m, m)$ is the entire space $F_2^m$ with generator matrix $G(m,m)$.
}
\end{definition}

\begin{theorem}{\em(\cite[Theorem 1.10.1]{HuffmanPless})\label{thm-RM code}}
Let $r$ be an integer with $0 \le r \le m$. Then the following hold:
\begin{enumerate}

\item[{\rm {(i)}}]
 $\mathcal R(i, m) \subseteq \mathcal R(j, m)$, if $0 \le i \le j \le m$.

\item[{\rm {(ii)}}]
The dimension of $\mathcal R(r, m)$ equals
\[
\binom{m}{0}+ \binom{m}{1}+ \cdots + \binom{m}{r}.
\]

\item[{\rm{(iii)}}]
 The minimum weight of $\mathcal R(r, m)$ equals $2^{m-r}$.

\item[{\rm{(iv)}}]
$\mathcal R(m, m)^{\perp} = \{{\bf 0} \}$, and if $0 \le r < m$, then $\mathcal R(r, m)^{\perp}
 = \mathcal R(m-r-1, m)$.
\end{enumerate}
\end{theorem}

\begin{theorem}\label{thm-RM}
Let $\mathcal R(r,m)$ be a binary $r$th order Reed-Muller code of length $2^m$.
\begin{enumerate}
\item[{\rm{(i)}}]
If $m$ is an odd integer $\ge 3$, then there is a fuzzy self-dual code $A$ over $\mathbb F_2$ such that $A_1 =\{ {\bf 0} \}$ and $A_{\alpha_r} = \mathcal R(r,m)$ for $0 \le r \le \frac{m-3}{2}$.

\item[{\rm{(ii)}}] If $m$ is an even integer $\ge 2$, then there is a fuzzy self-orthogonal code $B$ over $\mathbb F_2$ such that $B_1 =\{ {\bf 0} \}$ and $B_{\alpha_r} = \mathcal  R(r,m)$ for $0 \le r \le \frac{m}{2} -1$.
\end{enumerate}

\end{theorem}

\begin{proof}
By Theorem~\ref{thm-RM code}, $\mathcal  R(r,m)$ satisfies that $\mathcal  R(i, m) \subsetneq \mathcal  R(j,m)$ for $0 \le i < j \le m$, $\mathcal  R(m,m)^{\perp} =\{{\bf 0} \}$, $\mathcal R(m,m)=\mathbb F_2^n$, and
$\mathcal R(r,m)^{\perp}=\mathcal R(m-r-1, m)$ for $0 \le r < m$.

Suppose that $m$ is odd and that $0 \le r \le m$. Then $\mathcal R(r=\frac{m-1}{2}, m)$ is a self-dual code. Define $A(x)=1/2$ if $x \in \mathcal R(\frac{m-1}{2}, m)- \mathcal R(\frac{m-1}{2}-1, m)$ so that $A_{1/2}=\mathcal R(\frac{m-1}{2}, m)$.


Suppose that $1> \alpha_0 > \alpha_1 > \cdots > \alpha_{\frac{m-3}{2}} > \alpha_{\frac{m-1}{2}}=1/2$.
If $0 \le r \le \frac{m-3}{2}$, define $A_{\alpha_r}=\mathcal R(r,m)$ and $A_{1-\alpha_r}=\mathcal R(m-r-1, m)$. Then for $0 \le r \le \frac{m-3}{2}$, we have $A_{\alpha_r}^{\perp}=\mathcal R(r,m)^{\perp}=\mathcal R(m-r-1, m)=A_{1-\alpha_r}$.
Define $A_1=\{ {\bf 0} \}$ and $A_0=\mathbb F_2^n=\mathcal R(m,m)$ so that $A_1^{\perp}= A_0$.
Therefore, $A$ is a fuzzy self-dual code over $\mathbb F_2$.

Now suppose that $m$ is even and that $0 \le r \le m$.
Suppose that $1> \alpha_0 > \alpha_1 > \cdots > \alpha_{\frac{m}{2} -1} > 1/2$.
If $0 \le r \le \frac{m}{2} -1$, define $B_{\alpha_r}=\mathcal R(r,m)$ and $B_{1-\alpha_r}=\mathcal R(m-r-1, m)$.
Then for $0 \le r \le \frac{m}{2} -1$, $B_{\alpha_r}^{\perp}=\mathcal R(r,m)^{\perp}=\mathcal R(m-r-1, m)=B_{1-\alpha_r}$.
Define $B_1=\{ {\bf 0} \}$ and $B_0=\mathbb F_2^n=\mathcal R(m,m)$ so that $B_1^{\perp}= B_0$.
Therefore, $B$ is a fuzzy self-orthogonal code over $\mathbb F_2$.
\end{proof}

\begin{example}{\em
Consider the binary 2nd order Reed-Muller $[32,16,8]$ code $\mathcal R(2,5)$, which is self-dual, that is, $\mathcal R(2,5)=\mathcal R(2,5)^{\perp}$.
$\mathcal R(r,5)$ has dimension $\sum_{i=0}^r {5 \choose i}$.
We construct a dimension-parameterized fuzzy self-dual code containing all the Reed-Muller codes $\mathcal R(r,5)$ as follows.

Define
$A: V_{32} \rightarrow [0,1]
$
such that
\[A(\x) = \left\{
\begin{array}{ll}
		1  & \mbox{ if } x \in S_0=\{{\bf 0} \} \\
	   1-1/32 & \mbox{ if } x \in \mathcal R(0,5) - S_0 \\
	   1-6/32 & \mbox{ if } x \in \mathcal R(1,5) - \mathcal R(0,5) \\
	   1/2 & \mbox{ if } x \in \mathcal R(2,5) - \mathcal R(1,5) \\
	   6/32 & \mbox{ if } x \in \mathcal R(3,5) - \mathcal R(2,5) =\mathcal R(1,5)^{\perp} - \mathcal R(2,5) \\
	   1/32 & \mbox{ if } x \in  \mathcal R(4,5) - \mathcal R(3,5) =\mathcal R(0,5)^{\perp} - \mathcal R(1,5)^{\perp} \\
	   0 & \mbox{ if } x \in  \mathcal R(5,5) - \mathcal R(4,5) =S_0^{\perp} - \mathcal R(0,5)^{\perp}. \\
	\end{array}
\right.
\]
By Theorem~\ref{thm-RM}, $A$ gives a fuzzy self-dual code.
}
\end{example}

\section{Decoding fuzzy linear codes}
In this section, we describe how to decode fuzzy linear codes.

A basic idea of decoding fuzzy linear codes was given by~\cite{ShuGan}. Their idea is to decode the received vector with incorrect membership as a correct codeword with correct membership as follows.

Suppose that ${\bf x}$ is a codeword of a fuzzy linear code $A$ and that ${\bf y}={\bf x} + {\bf e}$, where ${\bf e}$ is an error vector, is a received vector. The authors~\cite{ShuGan} assume that $A({\bf x}) \ge A({\bf y})$.

Step 1) Let $\beta =A({\bf y})$ and get the corresponding linear code $A_{\beta}$.

Step 2) Use a decoding algorithm for $A_{\beta}$ to decode ${\bf y}$.

\medskip

There is no more discussion with some explicit example. In fact, Step 2) is not correct because ${\bf y}$ is already in $A_{\beta}$ and so ${\bf y}$ has no error in $A_{\beta}$.

\medskip

Hence we need to describe the above steps more explicitly as follows.

\begin{enumerate}

\item[{(i)}]
Suppose that we send only ${\bf x}$ with $A({\bf x}) \ge \alpha_1$ for some $\alpha_1 \in (0, 1]$ although it is common to consider $\alpha_1=1/2$. Hence, $A_{\alpha_1}$ is a linear $[n, k_1 ,d_1]$ code over $\F_q$ such that ${\bf x} \in A_{\alpha_1}$. Let $\mathcal C = A_{\alpha_1}$.  Let $G_{\alpha_1}$ (resp. $H_{\alpha_1}$) be generator matrix (resp. parity check matrix) for $A_{\alpha_1}$.

\item[{(ii)}]
Suppose that ${\bf y}={\bf x} + {\bf e}$ is a received vector, where ${\bf e}$ is a possible error vector.
If $A({\bf y}) \ge \alpha_1$, then ${\bf y} \in \C$, that is, ${\bf y} ={\bf x}$. We are done.

If $A({\bf y}) < \alpha_1$, then we know that ${\bf y} \notin \C$, that is, ${\bf e} \ne {\bf 0}$. So, $A({\bf x}) > A({\bf y})$.

\item[{(iii)}]
  Let $\alpha_2=A({\bf y})$. Then $A_{\alpha_2}$ is an $[n, k_2, d_2]$ code over $\F_q$.
   Let $G_{\alpha_2}$ (resp. $H_{\alpha_2}$) be generator matrix (resp. parity check matrix) for $A_{\alpha_2}$.
  Since $k_2 > k_1$, we have $A_{\alpha_1} \subsetneq A_{\alpha_2}$. Note $G_{\alpha_2}= \left[\begin{array}{c}
  G_3 \\
  G_{\alpha_2}\\
  \end{array}
  \right]$, where $G_3$ is a $(k_2-k_1) \times n$ matrix.

\item[{(iv)}]
 Consider the quotient space $A_{\alpha_2}/A_{\alpha_1}$ over $\F_q$, which has dimension $k_2 -k_1$.

 Compute the syndrome of ${\bf y}$ with respect to $H_{\alpha_1}$, that is, ${\bf s}_1 =H_{\alpha_1}{\bf y}^T$.

We generate a complete list of syndrome ${\bf s}_3= H_{\alpha_1}{\bf y}_3^T$ of any linear combination ${\bf y}_3$ of the rows of $G_3$ in Step (iii).

Then by following the Syndrome Decoding
Algorithm in~\cite[p. 42]{HuffmanPless}, create a table pairing the syndrome with the coset leader ${\bf e}_1$, which is, by definition, a vector of the smallest weight in the coset ${\bf e}_1 + \mathcal C$ of $A_{\alpha_2}/A_{\alpha_1}$.


\item[{(v)}]
If ${\bf s}_1 = H_{\alpha_1}{\bf y}^T$ is paired with
some coset leader ${\bf e}_1$, that is, $H_{\alpha_1}{\bf y}^T= H_{\alpha_1}{\bf e}_1^T$,
 then ${\bf y} - {\bf e}_1 \in \C$. Therefore, {\bf y} is corrected as ${\bf x}={\bf y} - {\bf e}_1$ and
$A({\bf y})$ is corrected as $A({\bf x})$. This completes the decoding algorithm.
\end{enumerate}

We remark that the table in Step (iv) has $q^{k_2-k_1
}$ entries, which is much smaller than the usual syndrome table of size $q^{n-k_1}$. This is an advantage of decoding by a fuzzy linear code. The Syndrome decoding can be simplified if the linear code $\C$ has an efficient decoding algorithm.


\medskip

\begin{example}
{\em
We describe this decoding algorithm to the fuzzy Hamming  $[8,4,4]$ code $\overline{\mathcal H}_3$ in Example~\ref{eg-code-D}.
From Equation (\ref{eq-Ham-8-4-4}), it has parity check matrix

\[
G(\overline{\mathcal H}_3)=\left[\begin{array}{cccccccc}
1 & 0 & 0 & 0& 0& 1& 1 & 1 \\
0 & 1 & 0 & 0& 1& 0& 1 & 1 \\
0 & 0 & 1 & 0& 1& 1& 0 & 1 \\
0 & 0 & 0 & 1& 1& 1& 1 & 0 \\
\end{array}
\right]
=\left[\begin{array}{c}
{\bf r}_1 \\
{\bf r}_2 \\
{\bf r}_3 \\
{\bf r}_4 \\
\end{array}
\right].
\]

Let ${\bf x}=(0 0 1 0 1 1 0 1)={\bf r}_3$ with $A({\bf x})=5/8$ be an original codeword and assume a received vector ${\bf y}={\bf x} + {\bf e}=(0 0 1 1 1 1 0 1)$ with $A({\bf y})=3/8$. Since $\overline{\mathcal H}_3$ is self-dual, its parity check matrix $H_4$ is also $G(\overline{\mathcal H}_3)$.

 Then $H_4{\bf y}^T=(0 0 0 1)^T \ne (0 0 0 0)^T$. Since $A({\bf y})=3/8$, ${\bf y}$ is in $S_3^{\perp} - \overline{\mathcal H}_3$ using the notation of Example~\ref{eg-code-D}. We can easily check that the generator matrix for $S_3^{\perp}=A_{3/7}$ is

\[
 G(A_{3/7})=\left[\begin{array}{c}
0 0 0 1 0 0 0  0 \\
\hline
  H_4 \\
\end{array}
\right].
\]

 Hence, by letting ${\bf e}=(0 0 0 1 0 0 0)$, we compute its syndrome $H_4{\bf e}^T=(0 0 0 1)^T$. Therefore, by Step (v) of the above algorithm, we have $H_4{\bf y}^T= H_4{\bf e}^T$, that is, $H_4 ({\bf y} - {\bf e})^T=(0 0 0 0)^T$. Thus, we decode ${\bf y}$ as ${\bf x}={\bf y} -{\bf e}=(00101101)$, and decode $A({\bf y})=3/8$ as $A({\bf x})=5/8$, as wanted.
}
\end{example}

\section{Conclusion}
 We have connected fuzzy sets with linear codes and introduced the arithmetic of fuzzy linear codes.
 It turns out that a fuzzy linear code is a parameterized way to describe a linear code.
For the first time, we have defined the concept of fuzzy self-dual and self-orthogonal codes. In particular, we have defined fuzzy self-dual Hamming codes, fuzzy self-dual Golay codes, and fuzzy Reed-Muller codes.
Our approach based on fuzzy linear codes gives an insight by relating a given linear code to a parameterized sequence of its nested subcodes and nested supercodes. Furthermore, we have proposed the Syndrome decoding algorithm for fuzzy linear codes in general. This algorithm works more efficient when we consider fuzzy self-orthogonal codes.

\bigskip

\noindent
{\bf Declaration of competing interest}

There is no competing interest.

\bigskip

\noindent
{\bf Data availability}

No data was used for the research described in the article.

\end{document}